\documentclass[11pt]{article}
\usepackage{amsmath, amsthm, amscd, amsfonts, mathrsfs, amssymb, graphicx, enumerate}
\newtheorem{theorem}{Theorem}[section]
\newtheorem{lemma}{Lemma}[section]
\newtheorem{remark}{Remark}[section]

\newtheorem{corollary}{Corollary}[section]

\numberwithin{equation}{section}

\begin{document}
\title{Some new Karamata type inequalities and their applications to some entropies}
\author{Shigeru Furuichi$^1$, Hamid Reza Moradi$^2$, and Akram Zardadi$^3$\\
$^1${\small Department of Information Science, College of Humanities and Sciences, Nihon University,}\\
{\small 3-25-40, Sakurajyousui, Setagaya-ku, Tokyo, 156-8550, Japan.}\\
{\small E-mail: furuichi@chs.nihon-u.ac.jp}\\
$^2${\small Department of Mathematics, Payame Noor University (PNU), P.O.Box, 19395-4697, Tehran, Iran.}\\
{\small E-mail: hrmoradi@mshdiau.ac.ir}\\
$^3${\small Department of Mathematics, Payame Noor University (PNU), P.O.Box, 19395-4697, Tehran, Iran.}\\
{\small E-mail: azardadi1990@yahoo.com}
}
\date{}
\maketitle
{\bf Abstract.}
Some new inequalities of Karamata type are established with a convex function in this paper. The methods of our proof allow us to obtain an extended version of the reverse of Jensen inequality given by Pe\v cari\'c and  Mi\'ci\'c. 

Applying the obtained results, we give reverses for information inequality (Shannon inequality) in different types, namely ratio type and difference type, under some conditions. Also, we provide interesting inequalities for von Neumann entropy and quantum Tsallis entropy which is a parametric extension of von Neumann entropy. The inequality for von Neumann entropy recovers the non-negativity and gives a refinement for the weaker version of Fannes's inequality for only special cases. Finally, we estimate bounds for the Tsallis relative operator entropy.
\vspace{2mm}

\vspace{2mm}

\vspace{2mm}

\section{Introduction}
We start by recalling some well-known notions which will be used in the sequel.
Let $\mathbf{x}=\left( {{x}_{1}},\ldots ,{{x}_{n}} \right)$ and $\mathbf{y}=\left( {{y}_{1}},\ldots ,{{y}_{n}} \right)$ be two (finite) sequences of real numbers, and let ${{x}_{ 1 }}\ge \ldots \ge {{x}_{ n }}$, ${{y}_{ 1 }}\ge \ldots \ge {{y}_{ n }}$ denote the components of $\mathbf{x}$ and $\mathbf{y}$ in decreasing order,
respectively. The $n$-tuple $\mathbf{y}$ is said to majorize $\mathbf{x}$ (or $\mathbf{x}$ is to be majorized by $\mathbf{y}$) in symbols $\mathbf{x}\prec \mathbf{y}$, if 
\[\sum\limits_{i=1}^{k}{{{x}_{ i }}}\le \sum\limits_{i=1}^{k}{{{y}_{ i }}}\quad\text{ holds for }k=1,\ldots ,n-1,\text{ and }\quad\sum\limits_{i=1}^{n}{{{x}_{i}}}=\sum\limits_{i=1}^{n}{{{y}_{i}}}.\]
In what follows,  $\mathscr{H}$ means a complex Hilbert space with inner product $\left\langle \cdot,\cdot \right\rangle $  and $\mathbb{B}\left( \mathscr{H} \right)$ is the algebra of all bounded linear operators on $\mathscr{H}$. We denote by $\mathbb{B}{{\left( \mathscr{H} \right)}_{h}}$ the real subspace of all self-adjoint operators on $\mathscr{H}$ and by $\mathbb{B}{{\left( \mathscr{H} \right)}_{+}}$ the set of all positive invertible operators in $\mathbb{B}{{\left( \mathscr{H} \right)}_{h}}$. 
Here the symbol ${{\mathbf{1}}_{\mathscr{H}}}$ denotes the identity operator on $\mathscr{H}$. We write $A\ge 0$  to mean that the operator $A$ is positive. If $A-B\ge 0$,
then we write $A\ge B$. A linear map $\Phi :\mathbb{B}\left( \mathscr{H} \right)\to \mathbb{B}\left( \mathscr{K} \right)$ is positive if $\Phi \left( A \right)\ge 0$  whenever $A\ge 0$. It is said to be unital if $\Phi \left( {{\mathbf{1}}_{\mathscr{H}}} \right)={{\mathbf{1}}_{\mathscr{K}}}$. For a real-valued function $f$ of a real variable and a self-adjoint operator $A\in \mathbb{B}\left( \mathscr{H} \right)_h$, the value $f\left( A \right)$  is understood by means of the functional
calculus for self-adjoint operators. We use the symbol $J$ as an interval on $\mathbb{R}$ in the sequel.

We start with an elegant result as part of the motivation for this paper.
\medskip

\begin{theorem} {\bf (Karamata's inequality \cite{3})} Let $f:J\to \mathbb{R}$ be a convex function and $\mathbf{x}=\left( {{x}_{1}},\ldots ,{{x}_{n}} \right)$,  $\mathbf{y}=\left( {{y}_{1}},\ldots ,{{y}_{n}} \right)$ be two $n$-tuples such that ${{x}_{i}},{{y}_{i}}\in J$ ($i=1,\ldots ,n$). Then
\[\mathbf{x}\prec \mathbf{y}\text{ }\Leftrightarrow \text{ }\sum\limits_{i=1}^{n}{f\left( {{x}_{i}} \right)}\le \sum\limits_{i=1}^{n}{f\left( {{y}_{i}} \right)}.\]
\end{theorem}

The following extension of majorization theorem is due to Fuchs \cite{4}:

\medskip

\begin{theorem}{\bf (\cite{4})}\label{theorem_b}
 Let $f:J\to \mathbb{R}$ be a convex function, $\mathbf{x}=\left( {{x}_{1}},\ldots ,{{x}_{n}} \right)$, $\mathbf{y}=\left( {{y}_{1}},\ldots ,{{y}_{n}} \right)$ two decreasing $n$-tuples such that ${{x}_{i}},{{y}_{i}}\in J$ ($i=1,\ldots ,n$), and $\mathbf{p}=\left( {{p}_{1}},\ldots ,{{p}_{n}} \right)$ be a real $n$-tuple such that 
\begin{equation}\label{8c}
\sum\limits_{i=1}^{k}{{{p}_{i}}{{x}_{i}}}\le \sum\limits_{i=1}^{k}{{{p}_{i}}{{y}_{i}}}\quad\text{ for }k=1,\ldots ,n-1,\text{ and }\quad\sum\limits_{i=1}^{n}{{{p}_{i}}{{x}_{i}}}=\sum\limits_{i=1}^{n}{{{p}_{i}}{{y}_{i}}}.
\end{equation}
Then
\begin{equation}\label{8}
\sum\limits_{i=1}^{n}{{{p}_{i}}f\left( {{x}_{i}} \right)}\le \sum\limits_{i=1}^{n}{{{p}_{i}}f\left( {{y}_{i}} \right)}.
\end{equation}
\end{theorem}
The conditions (\ref{8c}) are sometimes called $p$-majorization \cite{Cheng1977}. See also Section A in Chapter 14 in \cite{MO1979}.

As an application of Theorem \ref{theorem_b}, we have the following inequality related to $m$-th moment:
$$
\sum_{i=1}^n p_i(x_i-\overline{x})^m \leq \sum_{i=1}^np_i(y_i-\overline{y})^m,
$$
where $\overline{x} \equiv \sum_{i=1}^n p_i x_i$ and $\overline{y} \equiv \sum_{i=1}^n p_i y_i$. Indeed, this case satisfies the conditions (\ref{8c}) and the function $f(x) = (x-\overline{x})^m$ is convex for $m \geq 1$.

In this paper, we obtain a complementary inequality to Karamata's inequality.  It is an extension of the inequality due to Pe\v cari\'c and Mi\'ci\'c. Let $f$ be a convex function on the interval $\left[ m,M \right]$, $\left( {{A}_{1}},\ldots ,{{A}_{n}} \right)$, $\left( {{B}_{1}},\ldots ,{{B}_{n}} \right)$ be two $n$-tuples of self-adjoint operators with $m{{\mathbf{1}}_{\mathscr{H}}}\le {{A}_{i}},{{B}_{i}}\le M{{\mathbf{1}}_{\mathscr{H}}}$ ($i=1,\ldots ,n$) and ${{p}_{1}},\ldots ,{{p}_{n}}$ be positive scalars with $\sum\nolimits_{i=1}^{n}{{{p}_{i}}}=1$. We prove, among other inequalities, if $\sum\nolimits_{i=1}^{n}{{{p}_{i}}{{A}_{i}}}=\sum\nolimits_{i=1}^{n}{{{p}_{i}}{{B}_{i}}}$, then for a given $\alpha \ge 0$
\[\sum\limits_{i=1}^{n}{{{p}_{i}}f\left( {{A}_{i}} \right)}\le \beta +\alpha \sum\limits_{i=1}^{n}{{{p}_{i}}f\left( {{B}_{i}} \right)},\]
where $\beta =\underset{t\in \left[ m,M \right]}{\mathop{\max }}\,\left\{ {{a}_{f}}t+{{b}_{f}}-\alpha f\left( t \right) \right\}$ with the notations ${{a}_{f}}=\frac{f\left( M \right)-f\left( m \right)}{M-m}$ and ${{b}_{f}}=\frac{Mf\left( m \right)-mf\left( M \right)}{M-m}$ which will be used throughout this paper.
Some applications and remarks are given as well.
\section{Main Results}

In order to prove our main result we need the following lemma.
\begin{lemma}\label{14}
Let $f:J\to \mathbb{R}$ be a convex function, ${{A}_{i}}\in \mathbb{B}{{\left( \mathscr{H} \right)}_{h}}$ ($i=1,\ldots ,n$)  with the spectra  in $J$, and let ${{\Phi }_{i}}:\mathbb{B}\left( \mathscr{H} \right)\to \mathbb{B}\left( \mathscr{K} \right)$ ($i=1,\ldots ,n$) be  positive linear mappings such that $\sum\nolimits_{i=1}^{n}{{{\Phi }_{i}}\left( {{\mathbf{1}}_{\mathscr{H}}} \right)}={{\mathbf{1}}_{\mathscr{K}}}$. Then for any  $x\in \mathscr{K}$ with $\left\| x \right\|=1$,
\begin{equation}\label{04}
f\left( \left\langle \sum\limits_{i=1}^{n}{{{\Phi }_{i}}\left( {{A}_{i}} \right)}\;x,x \right\rangle  \right)\le \left\langle \sum\limits_{i=1}^{n}{{{\Phi }_{i}}\left( f\left( {{A}_{i}} \right) \right)}\;x,x \right\rangle.
\end{equation}	
\end{lemma}
\begin{proof}
It is well-known that if $f$ is a convex function on an interval $J$, then for each point $\left( s,f\left( s \right) \right)$, there exists a real number ${{C}_{s}}$ such that
\begin{equation}\label{01}
f\left( s \right)+{{C}_{s}}\left( t-s \right)\le f\left( t \right),\quad\text{ for all }t\in J.
\end{equation}
(Of course, if $f$ is differentiable at $s$, then ${{C}_{s}}=f'\left( s \right)$.)\\
Fix $s\in J$. Since $J$ contains the spectra of the $A_i$ for $ i=1,\ldots ,n$, we may replace $t$ in the above inequality by $A_i$, via a functional calculus to get
\[f\left( s \right){{\mathbf{1}}_{\mathscr{H}}}+{{C}_{s}}{{A}_{i}}-{{C}_{s}}s{{\mathbf{1}}_{\mathscr{H}}}\le f\left( {{A}_{i}} \right).\]
Applying the positive linear mappings ${{\Phi }_{i}}$ and summing on $i$ from $1$ to $n$, this implies
\begin{equation}\label{02}
f\left( s \right){{\mathbf{1}}_{\mathscr{K}}}+{{C}_{s}}\sum\limits_{i=1}^{n}{{{\Phi }_{i}}\left( {{A}_{i}} \right)}-{{C}_{s}}s{{\mathbf{1}}_{\mathscr{K}}}\le \sum\limits_{i=1}^{n}{{{\Phi }_{i}}\left( f\left( {{A}_{i}} \right) \right)}.
\end{equation}
The inequality \eqref{02} easily implies, for any $x\in \mathscr{K}$ with $\left\| x \right\|=1$,
\begin{equation}\label{03}
f\left( s \right)+{{C}_{s}}\left\langle \sum\limits_{i=1}^{n}{{{\Phi }_{i}}\left( {{A}_{i}} \right)}\;x,x \right\rangle -{{C}_{s}}\;s\le \left\langle \sum\limits_{i=1}^{n}{{{\Phi }_{i}}\left( f\left( {{A}_{i}} \right) \right)}\;x,x \right\rangle.
\end{equation}
On the other hand, since $\sum\nolimits_{i=1}^{n}{{{\Phi }_{i}}\left( {{\mathbf{1}}_{\mathscr{H}}} \right)}={{\mathbf{1}}_{\mathscr{K}}}$ we have $\left\langle \sum\nolimits_{i=1}^{n}{{{\Phi }_{i}}\left( {{A}_{i}} \right)}\;x,x \right\rangle \in J$ where $x\in \mathscr{K}$ with $\left\| x \right\|=1$. Therefore, we may replace $s$ by $\left\langle \sum\nolimits_{i=1}^{n}{{{\Phi }_{i}}\left( {{A}_{i}} \right)}\;x,x \right\rangle $  in \eqref{03}. This yields \eqref{04}.
\end{proof}
\begin{remark}
It is worth to remark that the inequality \eqref{04}  is the extension of \cite[Theorem 1]{dragomir}.
\end{remark}

We now state our first result.
\begin{theorem}\label{15}
Let $f:\left[ m,M \right]\to \mathbb{R}$ be a convex function, ${{A}_{i}},{{B}_{i}}\in \mathbb{B}{{\left( \mathscr{H} \right)}_{h}}$ ($i=1,\ldots ,n$)  with the spectra in $\left[ m,M \right]$,  and let ${{\Phi }_{i}}:\mathbb{B}\left( \mathscr{H} \right)\to \mathbb{B}\left( \mathscr{K} \right)$ ($i=1,\ldots ,n$) be  positive linear mappings such that $\sum\nolimits_{i=1}^{n}{{{\Phi }_{i}}\left( {{\mathbf{1}}_{\mathscr{H}}} \right)}={{\mathbf{1}}_{\mathscr{K}}}$. If $\,\,\sum\nolimits_{i=1}^{n}{{{\Phi }_{i}}\left( {{A}_{i}} \right)}=\sum\nolimits_{i=1}^{n}{{{\Phi }_{i}}\left( {{B}_{i}} \right)}$, then for a given $\alpha \ge 0$
\[\sum\limits_{i=1}^{n}{{{\Phi }_{i}}\left( f\left( {{A}_{i}} \right) \right)}\le \beta \mathbf{1}_\mathscr{K}+\alpha \sum\limits_{i=1}^{n}{{{\Phi }_{i}}\left( f\left( {{B}_{i}} \right) \right)},\]
where
\begin{equation}\label{1}
\beta =\underset{t\in \left[ m,M \right]}{\mathop{\max }}\,\left\{ {{a}_{f}}t+{{b}_{f}}-\alpha f\left( t \right) \right\}.
\end{equation}
\end{theorem}
\begin{proof}
Since $f$ is a convex function, for any $t\in \left[ m,M \right]$ we can write
\begin{equation}\label{6}
f\left( t \right)\le {{a}_{f}}t+{{b}_{f}}.
\end{equation}
It follows from the inequality \eqref{6} that
\[f\left( {{A}_{i}} \right)\le {{a}_{f}}{{A}_{i}}+{{b}_{f}}{{\mathbf{1}}_{\mathscr{H}}}.\]
Applying a positive linear mapping ${{\Phi }_{i}}$ and summing, we obtain
\[\sum\limits_{i=1}^{n}{{{\Phi }_{i}}\left( f\left( {{A}_{i}} \right) \right)}\le {{a}_{f}}\sum\limits_{i=1}^{n}{{{\Phi }_{i}}\left( {{A}_{i}} \right)}+{{b}_{f}}{{\mathbf{1}}_{\mathscr{K}}}.\]
So for any $x\in \mathscr{K}$ with $\left\| x \right\|=1$, we have
\[\left\langle \sum\limits_{i=1}^{n}{{{\Phi }_{i}}\left( f\left( {{A}_{i}} \right) \right)}x,x \right\rangle \le {{a}_{f}}\left\langle \sum\limits_{i=1}^{n}{{{\Phi }_{i}}\left( {{A}_{i}} \right)}x,x \right\rangle +{{b}_{f}}.\]
Whence for a given $\alpha \ge 0$,
\begin{align}
& \left\langle \sum\limits_{i=1}^{n}{{{\Phi }_{i}}\left( f\left( {{A}_{i}} \right) \right)}x,x \right\rangle -\alpha \left\langle \sum\limits_{i=1}^{n}{{{\Phi }_{i}}\left( f\left( {{B}_{i}} \right) \right)}x,x \right\rangle  \nonumber\\ 
& \le {{a}_{f}}\left\langle \sum\limits_{i=1}^{n}{{{\Phi }_{i}}\left( {{A}_{i}} \right)}x,x \right\rangle +{{b}_{f}}-\alpha f\left( \left\langle \sum\limits_{i=1}^{n}{{{\Phi }_{i}}\left( {{B}_{i}} \right)}x,x \right\rangle  \right) \label{12}\\ 
& = {{a}_{f}}\left\langle \sum\limits_{i=1}^{n}{{{\Phi }_{i}}\left( {{A}_{i}} \right)}x,x \right\rangle +{{b}_{f}}-\alpha f\left( \left\langle \sum\limits_{i=1}^{n}{{{\Phi }_{i}}\left( {{A}_{i}} \right)}x,x \right\rangle  \right) \label{13}\\ 
& \le \beta, \nonumber
\end{align}
where for the inequality \eqref{12} we have used Lemma \ref{14}, and the equality \eqref{13} follows from the fact that  $\sum\nolimits_{i=1}^{n}{{{\Phi }_{i}}\left( {{A}_{i}} \right)}=\sum\nolimits_{i=1}^{n}{{{\Phi }_{i}}\left( {{B}_{i}} \right)}$. Consequently,
\[\left\langle \sum\limits_{i=1}^{n}{{{\Phi }_{i}}\left( f\left( {{A}_{i}} \right) \right)}x,x \right\rangle \le \beta +\alpha \left\langle \sum\limits_{i=1}^{n}{{{\Phi }_{i}}\left( f\left( {{B}_{i}} \right) \right)}x,x \right\rangle \]
for any $x\in \mathscr{K}$ with $\left\| x \right\|=1$.
\end{proof}
\begin{remark}
If $f:\left[ m,M \right]\to \mathbb{R}$ is concave function, then the reverse inequality is valid in Theorem \ref{15} with  $\beta =\underset{t\in \left[ m,M \right]}{\mathop{\min }}\,\left\{ {{a}_{f}}t+{{b}_{f}}-\alpha f\left( t \right) \right\}$. 
\end{remark}

\begin{remark}\label{remark25}
In Theorem \ref{15}, we put $m=0$, $M=1$, $f(t)=t \log t$, $n=1$ and $\Phi_i = \frac{1}{\dim \mathscr{H}}\mathrm{Tr}$.  $\mathrm{Tr}:\mathbb{B}\left( \mathscr{H} \right)\to\mathbb{R}$ is a usual trace.
Since $\lim_{t\to 0+} t \log t =0$, we use the usual convention $f(0)=0$ in standard information theory \cite{Cover}. Then, we have $a_{t\log t} = b_{t\log t}=0$ and $\beta =\max_{0<t \leq 1}\left(-\alpha t\log t\right) =\frac{\alpha}{e}$ by easy computations. Therefore, for two positive operators $A,B$ satisfying $\mathrm{Tr}[A]=\mathrm{Tr}[B]=1$ (then the condition $\sum\nolimits_{i=1}^{n}{{{\Phi }_{i}}\left( {{A}_{i}} \right)}=\sum\nolimits_{i=1}^{n}{{{\Phi }_{i}}\left( {{B}_{i}} \right)}$ is trivially satisfied), we have the following interesting inequality:
\begin{equation}\label{remark23_ineq01}
\alpha H(B) \leq H(A) +\frac{\alpha}{e} \dim \mathscr{H},\quad (\alpha \geq 0)
\end{equation}
where $H(X) \equiv -\mathrm{Tr}[X\log X]$ is von Neumann entropy (quantum mechanical entropy) \cite{Nil} for a self-adjoint positive operator $X$ with unit trace.
The inequality recovers the non-negativity $H(A) \geq 0$, which is a fundamental property of von Neumann entropy, by taking $\alpha =0$.
Also we obtain the inequality:
\begin{equation}\label{remark25_ineq02}
|H(A)-H(B)|\leq\frac{\dim \mathscr{H}}{e},
\end{equation}
by taking $\alpha =1$ in \eqref{remark23_ineq01} and performing a replacement $A$ and $B$.
It may be interesting to compare 
the inequality given \eqref{remark25_ineq02} and the weaker version of Fannes's inequality \cite[(11.45)]{Nil}:
\begin{equation}\label{remark25_ineq03}
|H(A)-H(B)|\leq \mathrm{Tr}[|A-B|] \log \dim \mathscr{H} +\frac{1}{e}.
\end{equation}
If the dimension of Hilbert space $\mathscr{H}$ is large，then the upper bound of (\ref{remark25_ineq03}) is trivially tighter than that of (\ref{remark25_ineq02}). When $\dim \mathscr{H}=1$, both upper bounds coincide. For example, for the simple case $\mathrm{Tr}[|A-B|]=1$,
the upper bound of (\ref{remark25_ineq02}) is tighter than that of (\ref{remark25_ineq03}) when $\dim \mathscr{H}\leq 5$, while 
the upper bound of (\ref{remark25_ineq03}) is tighter than that of (\ref{remark25_ineq02}) 
when $\dim \mathscr{H}\geq 6$. As a conclusion, the inequality (\ref{remark25_ineq03}) gives tighter upper bound than ours in almost cases. That is,  our inequality  (\ref{remark25_ineq02}) gives a refinement for the weaker version of Fannes's inequality for only special cases.
\end{remark}

The method given in Remark \ref{remark25} is applicable to a generalized function in the following.
\begin{remark}
We use same setting in Remark \ref{remark25} except for the function $f_r(t)=\frac{t-t^{1-r}}{r}$ for $t>0$ and $0<r \leq 1$. Note that $\lim_{r\to 0} f_r(t) = t\log t$ so $f_r(t)$ is a parametric generalization of the function $t\log t$ used in Remark \ref{remark25}. We easily find $f_r''(t)=(1-r)t^{-r-1} \geq 0$ and $a_{f_r} = b_{f_r}=0$. Then $\beta =\max_{0<t \leq 1} g_{r,\alpha}(t)$, where $g_{r,\alpha}(t) \equiv \frac{\alpha}{r}(t^{1-r}-t)$. By easy computations, we have $g_{r,\alpha}'(t)=\frac{\alpha}{r}\left\{(1-r)t^{-r}-1\right\}$ and $g_{r,\alpha}''(t)=-\alpha(1-r)t^{-r-1} \leq 0$. Thus $g_{r,\alpha}$ takes maximum at $t=(1-r)^{\frac{1}{r}}$ and then $\beta = g_{r,\alpha}((1-r)^{\frac{1}{r}})=\alpha(1-r)^{\frac{1-r}{r}}$. By Theorem \ref{15}, we thus have
\begin{equation}\label{ineq00_remark26}
\alpha H_r(B) \leq H_r(A)+\alpha(1-r)^{\frac{1-r}{r}}\dim \mathscr{H},\quad (\alpha \geq 0,\,\, 0<r \leq 1)
\end{equation}
where $H_r(X)\equiv \frac{1}{r} \mathrm{Tr}[X^{1-r}-X] = -\mathrm{Tr}[X^{1-r}\ln_r X]$ defined for a self-adjoint positive operator $X$ with a unit trace, is often called quantum Tsallis entropy. 
See \cite{F2006,F2008}, for example. Note that the function $\ln_r t \equiv \frac{t^r -1}{r}$ defined for $t>0$ is often called $r$-logarithmic function and uniformly converges to standard logarithmic function $\log t$ in the limit $r \to 0$. Therefore, the inequality \eqref{ineq00_remark26} recovers the inequality \eqref{remark23_ineq01} in the limit $r\to 0$,
since $\lim_{r\to 0}H_r(A) =H(A)$ and $\lim_{r\to 0}(1-r)^{\frac{1-r}{r}} =\frac{1}{e}$. Thus we have the non-negativity $H_r(A) \geq 0$ by taking $\alpha =0$ and the inequality:
\begin{equation}\label{ineq01_remark26}
|H_r(A)-H_r(B)| \leq (1-r)^{\frac{1-r}{r}}\dim \mathscr{H}
\end{equation}
 by taking $\alpha =1$ and performing a replacement $A$ and $B$.
The inequality \eqref{ineq01_remark26} recovers the inequality \eqref{remark25_ineq02}, taking the limit of $r \to 0$.
\end{remark}

\begin{corollary}\label{16}
Let $f:\left[ m,M \right]\to \mathbb{R}$ be a convex function, ${{A}_{i}},{{B}_{i}}\in \mathbb{B}{{\left( \mathscr{H} \right)}_{h}}$ ($i=1,\ldots ,n$)  with the spectra in $\left[ m,M \right]$, and ${{p}_{1}},\ldots ,{{p}_{n}}$ be positive scalars with $\sum\nolimits_{i=1}^{n}{{{p}_{i}}}=1$. If $\,\,\sum\nolimits_{i=1}^{n}{{{p}_{i}}{{A}_{i}}}=\sum\nolimits_{i=1}^{n}{{{p}_{i}}{{B}_{i}}}$, then for a given $\alpha \ge 0$
\[\sum\limits_{i=1}^{n}{{{p}_{i}}f\left( {{A}_{i}} \right)}\le \beta \mathbf{1}_\mathscr{H}+\alpha \sum\limits_{i=1}^{n}{{{p}_{i}}f\left( {{B}_{i}} \right)},\]
where $\beta$ is defined as in \eqref{1}.
\end{corollary}
\begin{proof}
We apply Theorem \ref{15} for positive linear mappings ${{\Phi }_{i}}:\mathbb{B}\left( \mathscr{H} \right)\to \mathbb{B}\left( \mathscr{H} \right)$ determined by ${{\Phi }_{i}}:X\mapsto {{p}_{i}}X$ ($i=1,\ldots ,n$).
\end{proof}
Apply Corollary \ref{16} to the function $f\left( t \right)={{t}^{r}}$ for $r\notin (0,1)$ and $f(t)=-t^r$ for $r \in (0,1)$, we have the following.

\begin{remark}
Let ${{A}_{i}},{{B}_{i}}\in \mathbb{B}{{\left( \mathscr{H} \right)}_{+}}$ ($i=1,\ldots ,n$) with the spectra in $\left[ m,M \right]$ and ${{p}_{1}},\ldots ,{{p}_{n}}$ be positive numbers with $\sum\nolimits_{i=1}^{n}{{{p}_{i}}}=1$. Put $h={M}/{m}\;$. If $\,\,\sum\nolimits_{i=1}^{n}{{{p}_{i}}{{A}_{i}}}=\sum\nolimits_{i=1}^{n}{{{p}_{i}}{{B}_{i}}}$, then for any $r\notin (0,1)$,
\begin{equation} \label{ex01}
	\sum\limits_{i=1}^{n}{{{p}_{i}}A_{i}^{r}}\le K\left( h,r \right)\sum\limits_{i=1}^{n}{{{p}_{i}}B_{i}^{r}},
\end{equation}
where the generalized Kantorovich constant \cite{pe} is defined by
\[K\left( h,r \right)=\frac{({{h}^{r}}-h)}{\left( r-1 \right)\left( h-1 \right)}{{\left( \frac{r-1}{r}\frac{{{h}^{r}}-1}{{{h}^{r}}-h} \right)}^{r}}.\]
In addition, for any $r\notin (0,1)$
\begin{equation} \label{ex02}
	\sum\limits_{i=1}^{n}{{{p}_{i}}A_{i}^{r}}\le C\left(h,r \right)+\sum\limits_{i=1}^{n}{{{p}_{i}}B_{i}^{r}}
\end{equation}
where
\[C\left(h,r \right)=m^r\left\{\frac{h-{{h}^{r}}}{h-1}+\left( r-1 \right){{\left( \frac{{{h}^{r}}-1}{r\left( h-1 \right)} \right)}^{\frac{r}{r-1}}}\right\}.\]
Similarly, we have for $r\in (0,1)$
\begin{equation} \label{ex02_added01}
	\sum\limits_{i=1}^{n}{{{p}_{i}}A_{i}^{r}}\ge K\left( h,r \right)\sum\limits_{i=1}^{n}{{{p}_{i}}B_{i}^{r}}, \quad 	\sum\limits_{i=1}^{n}{{{p}_{i}}A_{i}^{r}}\geq C\left(h,r \right)+\sum\limits_{i=1}^{n}{{{p}_{i}}B_{i}^{r}}.
\end{equation}

Therefore, if we set $n=1$ and $mZ \leq X,Y \leq MZ$ for self-adjoint positive operators $X,Y,Z$, then we have the following relations:
\begin{itemize}
\item For $r\notin (0,1)$, we have
$$
Z\natural_r X \leq K(h,r) Z\natural_rY,\quad Z\natural_r X \leq C(h,r)Z+ Z\natural_rY. 
$$
\item For $r\in [0,1]$, we have
$$
Z\sharp_r X \geq K(h,r) Z\sharp_rY,\quad Z\sharp_r X \geq C(h,r)Z+Z\sharp_rY.
$$
\end{itemize}
Where $X\natural_rY \equiv X^{1/2}(X^{-1/2}YX^{-1/2})X^{1/2}$ for $r\in \mathbb{R}$ and one use the standard symbol $\sharp_r$ for $r\in [0,1]$ and it is called a weighted geometric mean for two positive operators $X$ and $Y$, inserting
$A=Z^{-1/2}XZ^{-1/2}$ and $B=Z^{-1/2}YZ^{-1/2}$ in the inequalities \eqref{ex01}, \eqref{ex02} and \eqref{ex02_added01}.
Thus we have the following relations.
\begin{itemize}
\item[(i)] For $r \geq 1$, we have
$$S_r(Z||X) \leq \frac{K(h,r) Z\natural_rY-Z}{r},\quad S_r(Z||X) \leq \frac{C(h,r)}{r}Z  +S_r(Z||Y).$$
\item[(ii)] For $r < 0$,  we have 
$$S_r(Z||X) \geq \frac{K(h,r) Z\natural_rY-Z}{r},\quad S_r(Z||X) \geq \frac{C(h,r)}{r}Z  +S_r(Z||Y).$$
\item[(iii)] For $0<r<1$,  we have
$$
S_r(Z||X) \geq \frac{K(h,r) Z\sharp_rY-Z}{r},\quad S_r(Z||X) + \frac{C(h,r)}{r}Z  \geq S_r(Z||Y).
$$
\end{itemize}
Here $S_r(X||Y) \equiv \frac{X\natural_r Y -X}{r} =X^{1/2}\ln_r(X^{-1/2}YX^{-1/2})X^{1/2}$ defined for $r \in \mathbb{R}$ and  a self-adjoint positive operators $X$ and $Y$, is called Tsallis relative operator entropy.
See  \cite{FYK,FM2018}, for example.  

Taking the limit as $r \to 0$ in the inequalities (ii) and (iii) above, we have
$S_0(Z||X) \geq 0$ and $S_0(Z||X) +S_0(Z||Y) \geq Z$, since $X\natural_0Y =X=X\sharp_0Y$, $\lim_{r\to 0}K(h,r) =1$ and $\lim_{r\to 0}C(h,r)=0$. Where $S_0(X||Y) \equiv X^{1/2}\log(X^{-1/2}YX^{-1/2})X^{1/2}$ is called relative operator entropy and we have the relation $\lim_{r\to 0} S_r(X||Y) = S_0(X||Y)$ by $\lim_{r\to 0}\ln_r x = \log x$.  
\end{remark}

\begin{remark}
When ${{B}_{1}}={{B}_{2}}=\cdots ={{B}_{n}}=\sum\nolimits_{i=1}^{n}{{{p}_{i}}{{A}_{i}}}$, Corollary \ref{16} reduces to
\[\sum\limits_{i=1}^{n}{{{p}_{i}}f\left( {{A}_{i}} \right)}\le \beta \mathbf{1}_\mathscr{H}+\alpha f\left( \sum\limits_{i=1}^{n}{{{p}_{i}}{{A}_{i}}} \right)\]
which is well-known in \cite[Theorem 3.2]{pecaric}.
\end{remark}

A commutative version for Corollary \ref{16} is straightforwardly obtained in the following.
\begin{corollary}\label{5}
	Let $f:\left[ m,M \right]\to \mathbb{R}$ be a convex function, ${{x}_{i}},{{y}_{i}}\in \left[ m,M \right]$ ($i=1,\ldots ,n$), and ${{p}_{1}},\ldots ,{{p}_{n}}$ be positive numbers with $\sum\nolimits_{i=1}^{n}{{{p}_{i}}}=1$. If $\,\,\sum\nolimits_{i=1}^{n}{{{p}_{i}}{{x}_{i}}}=\sum\nolimits_{i=1}^{n}{{{p}_{i}}{{y}_{i}}}$, then for a given $\alpha \ge 0$
	\begin{equation}\label{ineq00_theorem21}
	\sum\limits_{i=1}^{n}{{{p}_{i}}f\left( {{y}_{i}} \right)}\le \beta +\alpha \sum\limits_{i=1}^{n}{{{p}_{i}}f\left( {{x}_{i}} \right)},	
	\end{equation}
where $\beta$ is defined as in \eqref{1}.
\end{corollary}

\begin{remark}\label{11}
	\hfill
	\begin{itemize}
		\item[(i)] 	It is worth emphasizing that we have not used the condition $\sum\nolimits_{i=1}^{k}{{{p}_{i}}{{x}_{i}}}\le \sum\nolimits_{i=1}^{k}{{{p}_{i}}{{y}_{i}}}$ ($k=1,\ldots ,n-1$) in Corollary \ref{5}.
		
		\item[(ii)] If one chooses ${{x}_{1}}={{x}_{2}}=\cdots ={{x}_{n}}=\sum\nolimits_{i=1}^{n}{{{p}_{i}}{{y}_{i}}}$  in Corollary \ref{5}, then we deduce
		\[\sum\limits_{i=1}^{n}{{{p}_{i}}f\left( {{y}_{i}} \right)}\le \beta +\alpha f\left( \sum\limits_{i=1}^{n}{{{p}_{i}}{{y}_{i}}} \right).\]
		Actually,  Corollary \ref{5} can be regarded as an extension of the reverse of scalar Jensen inequality.
	\end{itemize}	
\end{remark}

By choosing appropriate $\alpha $ and $\beta $, we obtain the following result.
\begin{corollary}\label{cor21}
	Let $f:\left[ m,M \right]\to \mathbb{R}$ be a convex function, ${{x}_{i}},{{y}_{i}}\in \left[ m,M \right]$ ($i=1,\ldots ,n$), and ${{p}_{1}},\ldots ,{{p}_{n}}$ be positive numbers with $\sum\nolimits_{i=1}^{n}{{{p}_{i}}}=1$. If $f\left( t \right)>0$ for all $t\in \left[ m,M \right]$ and $\sum\nolimits_{i=1}^{n}{{{p}_{i}}{{x}_{i}}}=\sum\nolimits_{i=1}^{n}{{{p}_{i}}{{y}_{i}}}$, then 
	\begin{equation}\label{cor21_ineq01}
	\sum\limits_{i=1}^{n}{{{p}_{i}}f\left( {{y}_{i}} \right)}\le K\left( m,M,f \right)\sum\limits_{i=1}^{n}{{{p}_{i}}f\left( {{x}_{i}} \right)},
	\end{equation}
	where $K\left( m,M,f \right)=\max \left\{ \frac{{{a}_{f}}t+{{b}_{f}}}{f\left( t \right)}:\text{ }t\in \left[ m,M \right] \right\}$. Additionally, 
	\begin{equation}\label{cor21_ineq02}
	\sum\limits_{i=1}^{n}{{{p}_{i}}f\left( {{y}_{i}} \right)}\le C\left( m,M,f \right)+\sum\limits_{i=1}^{n}{{{p}_{i}}f\left( {{x}_{i}} \right)},
	\end{equation}
	where $C\left( m,M,f \right)=\max \left\{ {{a}_{f}}t+{{b}_{f}}-f\left( t \right):\text{ }t\in \left[ m,M \right] \right\}$.
\end{corollary}

\begin{remark}\label{remark22}
	We relax the condition $\sum_{i=1}^n p_ix_i =\sum_{i=1}^n p_iy_i $ in Corollaries \ref{5} and \ref{cor21} to $\sum_{i=1}^n p_ix_i \leq \sum_{i=1}^n p_iy_i $. But we impose on the monotone decreasingness to the function $f$. We keep the other conditions as they are. Then we also have the inequalities \eqref{ineq00_theorem21},  \eqref{cor21_ineq01} and \eqref{cor21_ineq02}. 
	
	Let $0<\epsilon \ll 1$.	We here set $m=\epsilon$, $M=1$, $f(t)=-\log t$, $y_i=q_i$ and $x_i=p_i$ in Corollary \ref{cor21}.
	Then $K(\epsilon,1,-\log) = \max_{\epsilon \leq t < 1} \left(\frac{\log\epsilon}{\epsilon-1} \frac{t-1}{\log t}\right) =\frac{\log\epsilon}{\epsilon-1}>0$, since the function $\frac{t-1}{\log t}$ is monotone increasing on $t\in (0,1)$ and $\lim_{t\to 1}\frac{t-1}{\log t} =1$. We also find $C(\epsilon,1,-\log) =\max_{\epsilon\leq t\leq 1} g(t)$, where $g(t) \equiv \frac{\log \epsilon}{\epsilon -1}(1-t)+\log t$. By easy computations, we have $g'(t)=\frac{1}{t}-\frac{\log \epsilon}{\epsilon -1}$ and $g''(t)=-t^{-2}<0$, we thus find $g(t)$ takes a maximum value at $t=\frac{\epsilon -1}{\log \epsilon}$ and it is $C(\epsilon,1,-\log)=g(\frac{\epsilon -1}{\log \epsilon}) =-\log\frac{\log \epsilon}{\epsilon -1}+\frac{\log \epsilon}{\epsilon -1} -1 = \log S(\epsilon)$, where we used the property $S(h^{-1}) =S(h)$ for the Specht ratio $S(h) \equiv \frac{h^{\frac{1}{h-1}}}{e\log h^{\frac{1}{h-1}}}$ given in \cite{Spe} with $h=M/m$.
	
	Under the assumption $\sum_{i=1}^np_i^2 \leq \sum_{i=1}^np_iq_i$ with $\sum_{i=1}^nq_i=1$, we thus obtain the inequalities:
	\begin{equation*}\label{ineq01_remark22}
	-\left(\frac{\epsilon -1}{\log \epsilon}\right) \sum_{i=1}^np_i\log q_i \leq -\sum_{i=1}^np_i\log p_i, 
	\quad
	-\sum_{i=1}^np_i\log q_i -\log S(\epsilon) \leq -\sum_{i=1}^np_i\log p_i .
	\end{equation*}
	Both inequalities give the reverses of information inequality (Shannon inequality) \cite{Cover}:
	\begin{equation} \label{inf_ineq}
	H({\bf p})\equiv -\sum_{i=1}^n p_i \log p_i \leq -\sum_{i=1}^n p_i \log q_i.
	\end{equation}
	Here $H({\bf p})$ is often called Shannon entropy (information entropy) where ${\bf p}=\left\{p_1,p_2,\ldots,p_n\right\}$.
	It may be noted that $\left\{p_1,p_2,p_3\right\}=\left\{\frac{1}{3},\frac{1}{3},\frac{1}{3}\right\}$ and
	$\left\{q_1,q_2.q_3\right\}=\left\{\frac{1}{6},\frac{1}{3},\frac{1}{2}\right\}$ satisfy the condition $\sum_{i=1}^3p_i^2 = \sum_{i=1}^3p_iq_i$, and also $\left\{p_1,p_2,p_3\right\}=\left\{\frac{1}{4},\frac{1}{4},\frac{1}{2}\right\}$ and $\left\{q_1,q_2.q_3\right\}=\left\{\frac{1}{10},\frac{1}{10},\frac{4}{5}\right\}$ satisfy the condition $\sum_{i=1}^3p_i^2 < \sum_{i=1}^3p_iq_i$,
	for example.
	
	Similarly, we set $m=\epsilon$, $M=1$, $f(t)=-\log t$, $y_i=p_i$ and $x_i=q_i$ in Corollary \ref{cor21}. Then we obtain the following inequalities:
	\begin{equation}\label{ineq02_remark22}
	-\sum_{i=1}^np_i\log p_i \leq -\left(\frac{\log \epsilon}{\epsilon -1}\right) \sum_{i=1}^n p_i \log q_i,\quad -\sum_{i=1}^np_i\log p_i \leq \log S(\epsilon)- \sum_{i=1}^n p_i \log q_i,
	\end{equation}
	under the assumption $\sum_{i=1}^np_iq_i \leq \sum_{i=1}^np_i^2$.
	Since $\frac{\log \epsilon}{\epsilon -1} >1$ for $0 <\epsilon < 1$ and $\log S(\epsilon)$ is decreasing on $0< \epsilon <1$ and $\log S(\epsilon) >0$ with $\lim_{\epsilon \to 0}\log S(\epsilon) =\infty$ and $\lim_{\epsilon \to 1}\log S(\epsilon) =0$, both inequalities \eqref{ineq02_remark22} above do not refine the  information inequality given in \eqref{inf_ineq}.
\end{remark}

\section{Conclusion and discussion}
We studied Karamata type inequality in both commutative case and non-commutative case.
Firstly, in commutative case, we gave the generalized Karamata type inequality by the use of convexity without the condition $\sum_{i=1}^kp_ix_i \leq \sum_{i=1}^kp_iy_i$ for $k=1,2,\ldots,n-1$ as one of main results. As a corollary, we obtained two different type inequalities. By relaxing the condition $\sum_{i=1}^np_ix_i = \sum_{i=1}^np_iy_i$ to $\sum_{i=1}^np_ix_i \leq \sum_{i=1}^np_iy_i$ but imposing the monotone decreasingness on the function $f$, then we obtained two different type inequalities for information inequality (Shannon inequality).
We actually showed the existence the probability distributions satisfying the condition $\sum_{i=1}^n p_i^2 \leq \sum_{i=1}^np_iq_i$.
Secondly in non-commutative case, we gave Karamata type inequality for a self-adjoint operators and positive linear mappings by the use of convexity, as one of main results.  As applications of this result, we gave interesting inequalities for von Neumann entropy and quantum Tsallis entropy. These inequalities interpolate the non-negativity for their entropies and their refinements for the weaker version of Fannes's type inequality. 
Finally, we estimated the bounds for Tsallis relative operator entropy.

Concluding this article, we try to obtain parametric extended results for Remark \ref{remark22},   by the similar application to a generalized function.
It is known that we have $r$-parametric extended information inequality in the form \cite{F2004}:
$$
H_r({\bf p}) \equiv -\sum_{i=1}^n p_i^{1-r} \ln_r p_i\leq - \sum_{i=1}^n p_i^{1-r} \ln_r q_i
$$
which can be proven by
$$
0= -\ln_r\sum_{i=1}^n\left(p_i\frac{q_i}{p_i}\right)  \leq -\sum_{i=1}^n p_i\ln_r\frac{q_i}{p_i}
=\sum_{i=1}^n\frac{p_i-p_i^{1-r}q_i^r}{r} = -\sum_{i=1}^np_i^{1-r}\ln_rq_i +\sum_{i=1}^np_i^{1-r}\ln_rp_i.$$

By simple computations, we have Tsallis relative entropy can be deformed as 
$$ H_r({\bf p}) = -\sum_{i=1}^n p_i^{1-r} \ln_r p_i=\sum_{i=1}^np_i\ln_r\frac{1}{p_i}.$$
However the following inequality 
\begin{equation} \label{conclusion_ineq01}
\sum_{i=1}^n p_i \ln_r \frac{1}{p_i}\leq  \sum_{i=1}^np_i \ln_r\frac{1}{q_i}
\end{equation}
does not hold in general, since $\sum_{i=1}^np_i \ln_r\frac{1}{q_i} \neq  -\sum_{i=1}^n p_i^{1-r} \ln_r q_i$ for $r>0$, in general. 

Therefore it is natural to find constants $c_1 > 0$ and $c_2 > 0$ such that
$$
\sum_{i=1}^np_i\ln_r\frac{1}{p_i} \leq c_1 \sum_{i=1}^np_i \ln_r\frac{1}{q_i},\quad \sum_{i=1}^np_i\ln_r\frac{1}{p_i} \leq c_2 +\sum_{i=1}^np_i \ln_r\frac{1}{q_i}.
$$
Actually, we can analyze our aimed inequality such as
$$
\sum_{i=1}^n p_i \ln_r\frac{1}{p_i} =-\sum_{i=1}^n p_i^{1-r} \ln_r p_i \leq -\sum_{i=1}^n p_i^{1-r} \ln_r q_i \leq c_1 \sum_{i=1}^n p_i \ln_r \frac{1}{q_i}.
$$
The last inequality in the above holds if there exists $c_1 >0$ such that $q_i \leq c_1^{\frac{1}{r}} p_i$ for any $r>0$ and any $i=1,\ldots,n$. 
Similarly we can analyze our aimed inequality such as
$$
\sum_{i=1}^n p_i \ln_r\frac{1}{p_i} =-\sum_{i=1}^n p_i^{1-r} \ln_r p_i \leq -\sum_{i=1}^n p_i^{1-r} \ln_r q_i \leq c_2 + \sum_{i=1}^n p_i \ln_r \frac{1}{q_i}.
$$
The last inequality in the above holds if there exists $c_2$ such that $0< c_2 <\frac{1}{r}$ and $q_i \leq (1-rc_2)^{\frac{1}{r}} p_i$ for any $r>0$ and any $i=1,\ldots,n$. 
For these cases, we can not give the explicit forms for constants $c_1$ and $c_2$ and need a condition on two probability distributions. For the cases of $c_1$, it is necessary to satisfy $p_i > 0$, at least. For the case $c_2$, it is also necessary to satisfy $q_i>0$, at least. And the constants $c_1$ and $c_2$ have to vary to satisfy the corresponding conditions. They are no longer constants.
In order to give the explicit forms for constants $c_1$ and $c_2$ under a certain condition of the probability distributions ${\bf p}=\left\{p_1,\ldots,p_n\right\}$ and ${\bf q}=\left\{q_1,\ldots,q_n\right\}$, we apply the method in Remark \ref{remark22}.
We set $m=\epsilon$, $M=1$, $y_i = p_i$, $x_i = q_i$
 and the function $f_r(t) = \ln_r t^{-1} =\frac{t^{-r} -1}{r}$ for $0<t<1$ and $r>0$. Since the function $\ln_r t \equiv \frac{t^r -1}{r}$ uniformly converges to the usual logarithmic function $\log t$ in the limit $r \to 0$, $f_r(t)$ converges to $-\log t$ in the limit $r \to 0$. We easily find $f_r'(t)=-t^{-r-1} \leq 0$, $f_r''(t) = (r+1)t^{-r-2} \geq 0$, $a_{f_r} =\frac{1-\epsilon^{-r}}{r(1-\epsilon)}$ and $b_{f_r} =\frac{\epsilon^{-r}-1}{r(1-\epsilon)}$. Then
$K(\epsilon,1,f_r) =\max_{\epsilon \leq t <1}g_r(t)$, where
$g_r(t) \equiv \frac{(1-\epsilon^{-r})(t-1)}{(1-\epsilon)(t^{-r}-1)}$.
Then we calculate $g_r'(t) = \frac{(1-\epsilon^{-r})t^{r-1}h_r(t)}{(1-\epsilon)(t^{-r}-1)^2}$, where $h_r(t) \equiv r(t-1)-t(t^r-1)$ for $0<t<1$ and $r>0$. 
Since $h_r'(t) =r+1-(r+1)t^r$ and $h_r''(t)=-r(r+1)t^{r-1}<0$ for $r>0$,
we have $h_r'(t) \geq h_r'(1)=0$ which implies $h_r(t) \leq h_r(1)=0$.
According to the fact $\frac{1-\epsilon^{-r}}{1-\epsilon} \leq 0$ for $\epsilon \ll 1$ and $r >0$, we find $g_r'(t) \geq 0$. Thus we obtain
$K(\epsilon,1,f_r) =g_r(1) = \frac{\ln_r\epsilon^{-1}}{1-\epsilon} > 0$.
We also calculate $C(\epsilon,1,f_r) =\max_{\epsilon \le t <1} g_r(t)$, where $g_r(t)\equiv \frac{\ln_r \epsilon^{-1}}{1-\epsilon}(1-t)-\ln_rt^{-1}$.
We easily find $g_r'(t)=\frac{\ln_r \epsilon^{-1}}{\epsilon-1}+t^{-r-1}$ and $g_r''(t) =-(r+1)t^{-r-2} \leq 0$ so that $g_r(t)$ takes a maximum at $t=\left(\frac{1-\epsilon}{\ln_r\epsilon^{-1}}\right)^{\frac{1}{r+1}}$ and 
$$C(\epsilon,1,f_r) = g\left(\left(\frac{1-\epsilon}{\ln_r\epsilon^{-1}}\right)^{\frac{1}{r+1}}\right)=\frac{\ln_r\epsilon^{-1}}{1-\epsilon}-\left(\frac{\ln_r\epsilon^{-1}}{1-\epsilon}\right)^{\frac{r}{r+1}}-\ln_r\left(\frac{\ln_r\epsilon^{-1}}{1-\epsilon}\right)^{\frac{1}{r+1}} \equiv ls_r(\epsilon).$$
Thus we have the following inequalities with $c_1 = \frac{\ln_r\epsilon^{-1}}{1-\epsilon}$ and $c_2 = ls_r(\epsilon)$:
\begin{equation}\label{con_ineq02}
\sum_{i=1}^n p_i \ln_r \frac{1}{p_i}\leq
\left(\frac{\ln_r\epsilon^{-1}}{1-\epsilon}\right)\sum_{i=1}^np_i \ln_r\frac{1}{q_i}  ,\quad   \sum_{i=1}^n p_i \ln_r \frac{1}{p_i}\leq ls_r(\epsilon)+ \sum_{i=1}^np_i \ln_r\frac{1}{q_i},
\end{equation}
under the assumption $\sum_{i=1}^np_iq_i \leq \sum_{i=1}^np_i^2$. Thus we obtained two parametric extended inequalities (difference type and ratio type) under a certain assumption, with two constants $c_1$ and $c_2$, although the inequalities do not hold in general for the cases $c_1=1$ or $c_2=0$. This means our established mathematical tool, namely Corollary \ref{5} is applicable to obtain some results in natural science.
These inequalities \eqref{con_ineq02} recover the inequalities \eqref{ineq02_remark22}  in the limit $r \to 0$.
The reverse inequalities for a kind of the inequality \eqref{conclusion_ineq01} are similarly obtained in the following:
$$
\left(\frac{1-\epsilon}{\ln_r\epsilon^{-1}}\right)\sum_{i=1}^np_i \ln_r\frac{1}{q_i} \leq \sum_{i=1}^n p_i \ln_r \frac{1}{p_i},\quad  \sum_{i=1}^np_i \ln_r\frac{1}{q_i} -ls_r(\epsilon) \leq \sum_{i=1}^n p_i \ln_r \frac{1}{p_i},
$$
under the assumption $\sum_{i=1}^np_i^2 \leq \sum_{i=1}^np_iq_i$.

Finally, we study elementary properties on two constants $c_1$ and $c_2$. In the sequel, we often denote $c_1$ and $c_2$ for simplicity, although they are the functions of $\epsilon$ such as $c_1:=c_1(\epsilon)$ and $c_2:=c_2(c_1(\epsilon))$. Let $0<\epsilon <1$ and $r>0$. We firstly find that $\lim_{r\to 0}c_1 =\lim_{r\to 0}\frac{\ln_r\epsilon^{-1}}{1-\epsilon} =  \frac{\log \epsilon}{\epsilon-1} =K(\epsilon,1,-\log)$ and $\lim_{r\to 0}c_2 = \lim_{r\to 0}ls_r(\epsilon) = \log S(\epsilon) = C(\epsilon,1,-\log)$. The constant $ls_r(\epsilon)$ can be regarded as a parametric extension of $\log S(\epsilon)$.
We  find $\dfrac{dc_1(\epsilon)}{d\epsilon}
=\dfrac{\epsilon^{-r-1}\left\{-\epsilon^{r+1}+(r+1)\epsilon -r\right\}}{r(\epsilon^r-1)^2}\leq 0$ so that we have $c_1=c_1(\epsilon) > c_1(1)=\lim\limits_{\epsilon\to 1} =1$. The proof of $\frac{dc_1}{d\epsilon} \leq 0$  can be done by putting $k_r(\epsilon) \equiv -\epsilon^{r+1}+(r+1)\epsilon -r$ and calculating $k_r'(\epsilon) =(r+1)(1-\epsilon^r)\geq 0$ with $k_r(\epsilon) \leq k_r(1) =0$. Since $ls_r(\epsilon) = c_1 -c_1^{\frac{r}{r+1}}-\ln_rc_1^{\frac{1}{r+1}}$, we consider the function $l(c_1) = c_1 -c_1^{\frac{r}{r+1}}-\ln_rc_1^{\frac{1}{r+1}}$ for $c_1 >1$. Since $\frac{dl(c_1)}{dc_1} =1-c_1^{-\frac{1}{r+1}} \geq 0$ for $r>0$ and $ c_1 >1$.
Thus we have $l(c_1) \geq l(1) =0$ and we have $\lim\limits_{c_1\to 0}l(c_1) =\frac{1}{r}$ for $0<r < \infty$. Therefore we have $0 \leq ls_r(\epsilon) \leq \frac{1}{r}$.

\section*{Acknowledgement}
The authors would like to thank the referee for helping them to clarify the presentation of this paper. The author (S.F.) was partially supported by JSPS KAKENHI Grant Number 16K05257.

\end{document}